\documentclass[
reprint,
onecolumn,
aps,
pra,
longbibliography,
nofootinbib,
]{revtex4-2}

\usepackage[T2A]{fontenc}
\usepackage[utf8]{inputenc}
\usepackage[english,russian]{babel}

\usepackage{xcolor}      
\usepackage{graphicx}
\usepackage{dcolumn}
\usepackage{bm}
\usepackage{amsmath}
\usepackage{braket}
\usepackage{amsfonts}
\usepackage{bbm}
\usepackage{physics}
\usepackage{multirow}
\usepackage{tikz-cd}
\usepackage{tikz}
\usepackage{amsthm}

\usetikzlibrary{arrows.meta, bending}

\usepackage[most]{tcolorbox}

\newtheorem{theorem}{Theorem}[section]

\usepackage[unicode,colorlinks=true,citecolor=blue,urlcolor=blue]{hyperref}

\begin{document}

\selectlanguage{english}

\title{Spin--Boson Mappings in the Formalism of \texorpdfstring{$f$}{f}-Deformations}

\author{Vladimir A. Orlov}
\email{vorlovac@outlook.com}
\affiliation{Russian Quantum Center, Skolkovo, Moscow 121205, Russia}
\affiliation{Moscow Institute of Physics and Technology, Dolgoprudny, Moscow Region 141701, Russia}

\author{Liubov A. Markovich}
\affiliation{Russian Quantum Center, Skolkovo, Moscow 121205, Russia}

\author{Andrey V. Mikheyenkov}
\affiliation{Vereshchagin Institute for High Pressure Physics of the Russian Academy of Sciences, Troitsk, Moscow 108840, Russia}
\affiliation{Moscow Institute of Physics and Technology, Dolgoprudny, Moscow Region 141701, Russia}
\affiliation{National Research Center ``Kurchatov Institute'', Moscow 123182, Russia}

\author{Vladimir I. Man'ko}
\affiliation{Russian Quantum Center, Skolkovo, Moscow 121205, Russia}
\affiliation{P. N. Lebedev Physical Institute of the Russian Academy of Sciences, Moscow 119991, Russia}
\affiliation{Moscow Institute of Physics and Technology, Dolgoprudny, Moscow Region 141701, Russia}

\date{\today}

\begin{abstract}
We develop a unified algebraic approach to spin--boson transformations based on the formalism of $f$-deformed oscillators. In the single-mode case, we show that the Holstein--Primakoff and Dyson--Maleev transformations, together with the interpolating $\alpha$-family, arise as different factorizations of the same algebraically determined object. The standard spin--boson mappings can thus be interpreted as realizations of a common structure, which makes it possible to clearly separate the exact algebraic content on the physical subspace from effects associated with non-Hermiticity, the choice of metric, and extensions beyond the physical subspace. In the two-mode case, the same approach yields both deformed versions of the Jordan--Schwinger transformation and new exact two-mode realizations. Our results provide a unified description of known spin--boson transformations and naturally lead to new bosonic representations.
\end{abstract}

\maketitle
\renewcommand{\thefootnote}{\arabic{footnote})}

\selectlanguage{english}
\section{Introduction}

Spin--boson mappings (also referred to as spin--boson transformations) are important tools in quantum magnetism, quantum optics, many-body physics, and quantum computing~\cite{mandel1996optical, kolmer2024modeling, barberena2025generalized, omanakuttan2026holstein, dudinets2025circuit, bolsmann2023switching, hepp1973superradiant,wang1973phase,emary2003chaos,dimer2007proposed,dusuel2004finite,morrison2008dynamical,klein1991boson,ma2011quantum,kurucz2010multilevel, liu2026hybrid, descamps2024superselection, orlov2025discrete}. Their primary purpose is to map problems formulated in terms of spin operators obeying a nontrivial algebra onto equivalent problems in bosonic variables, for which a broad range of analytical and approximation techniques is available. This considerably simplifies the description of collective excitations, enables the use of quantum-field-theoretical methods, and facilitates the development of both analytical and numerical solution methods~\cite{Auerbach1994, mattis2012theory}.
\par The standard spin--boson transformations include the Holstein--Primakoff (HP), Dyson--Maleev (DM), and Jordan--Schwinger (JS) mappings~\cite{holstein1940field,dyson1956general,dyson1956thermodynamic,maleev1958scattering,Jordan1935,Schwinger1952}. Historically, these transformations were introduced independently to address different, albeit closely related, problems. In standard textbooks on magnetism and condensed-matter physics, the HP transformation is usually either stated without derivation or justified by a direct comparison of matrix elements~\cite{kittel2018introduction, white2007quantum, madelung2012introduction, ziman1979principles}. Such a treatment obscures the common algebraic mechanism by which spin--boson transformations emerge as particular realizations of a unified construction. In the original work~\cite{holstein1940field}, the HP transformation was constructed by exactly reproducing the matrix elements of spin operators in a Fock basis in which the boson occupation number counts deviations from the fully polarized state. Dyson's objective, by contrast, was to develop a convenient formalism for describing interacting spin waves in a ferromagnet~\cite{dyson1956general, dyson1956thermodynamic}. His approach replaces physical spin-wave states, which are generally nonorthogonal, with orthogonal bosonic states of an effective model. The spin operators retain their correct physical content, while their simpler polynomial representation in the enlarged bosonic space becomes non-Hermitian. Building on this construction, Maleev proposed a modified form of the transformation that is better suited to practical calculations~\cite{maleev1958scattering}.
\par The relationship among these approaches was not initially apparent and only became clear gradually. Moreover, at an early stage, the Dyson and Holstein--Primakoff methods yielded noticeably different results, leading one of the transformations to be regarded as incorrect. The first explanation of these discrepancies, which attributed them to an improper approximation of the relevant series, was given by Oguchi~\cite{oguchi1960theory}. An important step toward unification was made by Dembinski, who showed that spin--boson transformations must be considered together with the choice of a physical subspace and the corresponding metric~\cite{dembinski1964dyson}. Even within this framework, however, standard accounts typically list separate representations rather than derive them from a single algebraic mechanism.
\par The JS transformation occupies a special place in this family because it employs two independent bosonic modes~\cite{Jordan1935,Schwinger1952} and plays a central role in auxiliary-boson theories~\cite{arovas1988functional, kotliar1986new}. Subtle issues arise already at a basic level and are often left implicit in the applied literature: the relation between the finite-dimensional spin space and the infinite-dimensional Fock space, the behavior of the representations outside the physical subspace, and questions involving the metric and quasi-Hermiticity~\cite{scholtz1992quasi,mostafazadeh2001pseudo,mostafazadeh2010pseudo}. In practical calculations, moreover, not only the algebraic validity of a representation on the physical subspace but also its stability under expansions, truncations, and numerical implementation is essential.
\begin{figure}[h]
    \centering
    \includegraphics[width=0.3125\linewidth]{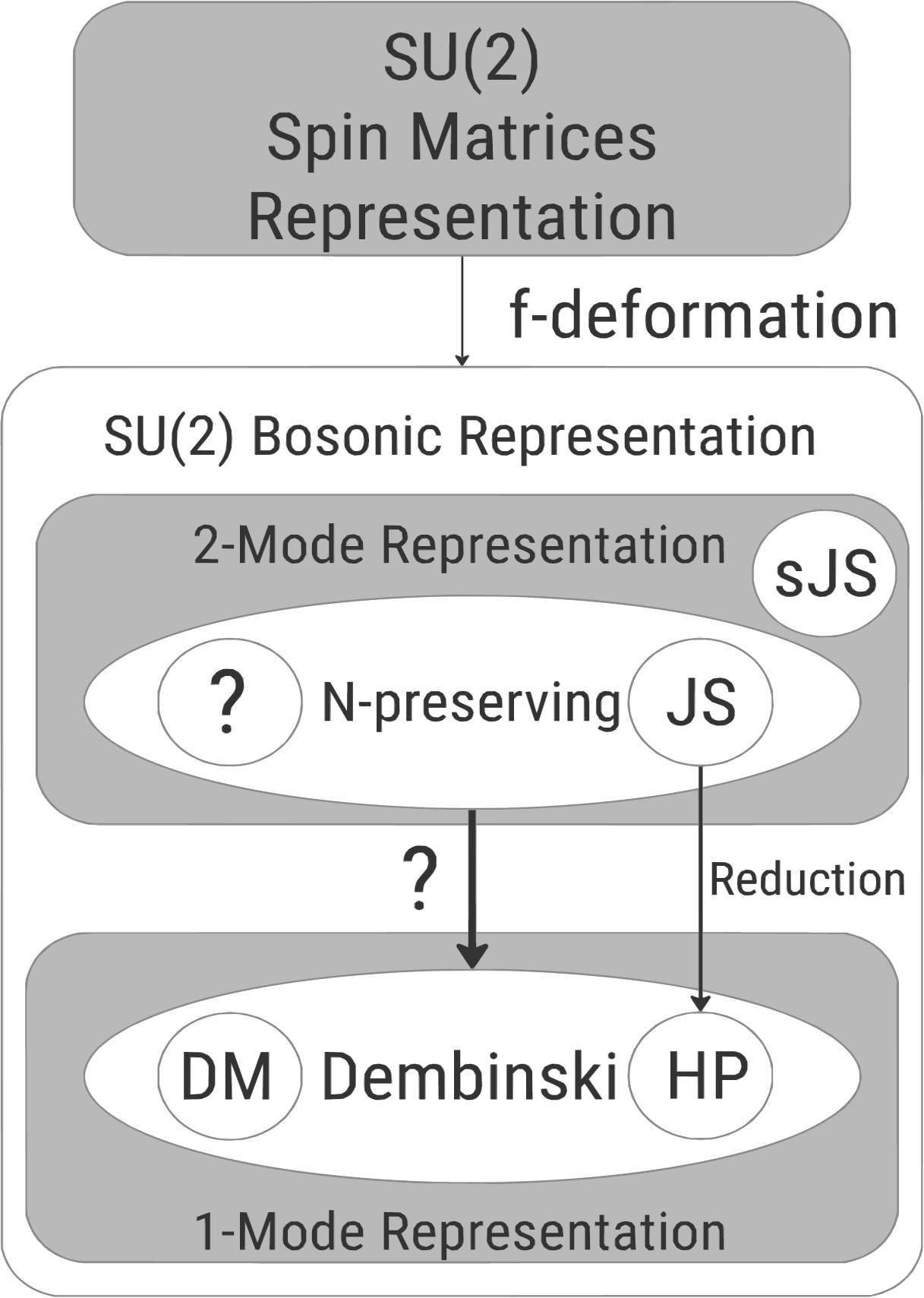}
\caption{Schematic overview of the general structure of spin--boson realizations of the $\mathfrak{su}(2)$ algebra within the $f$-deformation formalism.
The upper block corresponds to the standard matrix spin representation, from which the $f$-deformation approach generates a general class of bosonic realizations. This class contains both two-mode and single-mode representations. In the two-mode case, the diagram shows the standard $N$-conserving Jordan--Schwinger (JS) class, the squeezed Jordan--Schwinger transformation (sJS), and the possibility of more general realizations beyond the standard class. In the single-mode case, the Holstein--Primakoff (HP) and Dyson--Maleev (DM) transformations and the interpolating $\alpha$-family (the Dembinski family) arise as different factorizations of the same algebraically determined structure. The arrow labeled \emph{Reduction} denotes the reduction of the two-mode block to the single-mode physical subspace, while the question mark indicates additional factorization freedom and the possibility of new representations.
}
\label{fig:1}
\end{figure}
\par
The broader context is provided by generalized boson mappings, originating in the classic work of Belyaev and Zelevinsky~\cite{beliaev1962anharmonic} and systematically reviewed by Klein and Marshalek~\cite{klein1991boson}. Within this framework, one naturally asks which part of a spin--boson representation is determined by the $\mathfrak{su}(2)$ algebra itself and which part reflects the freedom to choose the factorization, normalization, metric, and extension beyond the physical subspace. This distinction is fundamental both to understanding the standard transformations and to constructing new bosonic realizations of other algebras. Thus, spin--boson transformations arose from different historical motivations, and their deep internal relationship became apparent only retrospectively. Existing approaches emphasize this relationship but do not provide a unified algebraic mechanism from which these transformations follow as special cases. The completeness of the known class of representations and the possibility of constructing new realizations therefore remain open questions. The formalism of $f$-deformed oscillators provides a natural language for addressing these questions: canonical bosonic operators are modified by functions of the particle-number operator~\cite{man1993physical, de1996nonlinear, man1998nonlinearity}. This framework generalizes the formalism of $q$-deformations~\cite{drinfel1985hopf, jimbo1985q, arik1976hilbert, biedenharn1989quantum, macfarlane1989q, kulish1983quantum, chang1992generalized, quesne2002new, csenay2024study, serdouk2026fractional, lotfizadeh2020construction, huang2025gauge, razumov2023poincare, Korybut2017TMP, BorzovDamaskinsky2011TMP, li2022symmetries, BorzovDamaskinsky2017TMP} and is widely used in mathematical physics and in the theory of nonlinear coherent states~\cite{kullock2016towards, chatterjee2024analyzing, man2002alternative, dudinets2017tomography, man2010moyal}. Nevertheless, its potential in the context of spin--boson transformations remains largely unexplored~\cite{momeni2022nonlinear}.
\par
The aim of this work is to provide a unified algebraic derivation of the standard spin--boson transformations from an $f$-deformed ansatz and thereby embed the known representations in a common framework (Fig.~\ref{fig:1}). We show that, in the single-mode case, the $\mathfrak{su}(2)$ commutation relations lead to a universal finite-difference equation for the invariant product of deformation functions, whose solution is uniquely fixed within the physical subspace. The standard transformations emerge as different factorizations of this object. This viewpoint identifies the $\alpha$-family of realizations, separates the exact block structure from its extensions beyond the physical subspace, and clarifies the role of non-Hermiticity and metric corrections.

We then extend the formalism to two-mode representations. We show that the JS transformation and its deformations admit an analogous description in terms of factorization freedom on a block with a fixed total boson number. We also consider realizations beyond the JS class, in which the total boson number need not be conserved even though the $\mathfrak{su}(2)$ algebra is preserved. A new class of deformed JS transformations is introduced. Thus, $f$-deformations provide both a unification of known representations and a systematic route to constructing new ones.

We pay particular attention to computational aspects. The accuracy of a calculation is determined not only by the choice of mapping but also by how operator functions of the discrete particle-number operator are treated. In this setting, finite-difference (Newton-series) expansions provide a natural tool: they respect the discrete structure of the physical subspace and allow normally ordered expressions to be controlled systematically~\cite{kubo1952spin,vogl_resummation_2020,konig2021newtonseriesexpansionbosonic}. This makes it possible to clearly distinguish properties fixed by the algebra from effects caused by the choice of representation, approximation, and extension beyond the physical subspace.

The paper is organized as follows. In Sec.~\ref{sec:one_mode}, we review known results for single-mode transformations and discuss their constraints, sources of error, and applications in detail. An analogous review of two-mode transformations is given in Sec.~\ref{sec:two_mode}. In Secs.~\ref{sec:fosc} and~\ref{sec:two_mode_js}, application of the $f$-deformation approach yields a general class of bosonic realizations within which single-mode and two-mode representations emerge naturally. In the single-mode case, this general class contains the HP and DM transformations and the interpolating $\alpha$-family, which are interpreted as different factorizations of the same algebraically determined object. In the two-mode case, we identify the standard $N$-conserving class that contains the JS transformation. Section~\ref{sec:beyond_js} demonstrates the existence of more general realizations that need not conserve the total boson number. We also introduce a new family, the squeezed Jordan--Schwinger transformation.

\selectlanguage{english}
\section{Single-Mode Transformations}
\label{sec:one_mode}
We now consider single-mode bosonic realizations of the $\mathfrak{su}(2)$ algebra, in which the spin operators are expressed in terms of a single bosonic mode.
Let $\hat S^\pm$ and $\hat S^z$ denote spin operators defined at a lattice site for an arbitrary spin $S\in\{\tfrac12,1,\dots\}$ that realize a representation of $\mathfrak{su}(2)$ with the commutation relations
\begin{align} \label{eq:su2_comm_ch3_1}
&[\hat S^z,\hat S^\pm]=\pm \hat S^\pm,\\
&[\hat S^+,\hat S^-]=2\hat S^z,
\label{eq:su2_comm_ch3}
\end{align}
Throughout this work, we use the conventional physics notation
$\mathfrak{su}(2)$ for the ladder-operator relations above.
Algebraically, these relations are those of the complexified algebra
$\mathfrak{su}(2)_{\mathbb C}\simeq\mathfrak{sl}(2,\mathbb C)$.
Hermiticity or quasi-Hermiticity conditions, when relevant for a
particular realization, are specified separately.

The standard Holstein--Primakoff (HP) transformation~\cite{holstein1940field} then reads
\begin{align}
\hat{S}^{+}&= \sqrt{2S}\left(1-\frac{\hat{n}}{2S}\right)^{\frac{1}{2}}\!\hat{a},\quad
 \hat{S}^{-}\!=\hat{a}^{\dagger}\sqrt{2S}\left(1-\frac{\hat{n}}{2S}\right)^{\frac{1}{2}}, \nonumber\\
 \hat{S}^{z}&=S-\hat{n},
\label{eq:HP_0}
\end{align}
where $\hat a$ and $\hat a^\dagger$ are the canonical bosonic annihilation and creation operators satisfying $[\hat a,\hat a^\dagger]=1$, and $\hat n=\hat a^\dagger\hat a$ is the boson-number operator. The Dyson--Maleev (DM) transformation~\cite{dyson1956thermodynamic, dyson1956general, maleev1958scattering} has the form
\begin{equation}
\hat{S}^{+}=\sqrt{2S}\left( 1-\frac{\hat{n}}{2S}\right)\!\hat{a},\quad\!\!\!\!
\hat{S}^{-}\!\!=\hat{a}^{\dagger }\sqrt{2S},\quad\! \hat{S}^{z}=S-\hat{n}.
\label{eq:DM_0}
\end{equation}
A dual form of the DM transformation is also used in the literature:
\begin{equation}
\hat{S}^{+}=\sqrt{2S}\hat{a},\quad\!\!\!\!
\hat{S}^{-}=\!\hat{a}^{\dagger }\sqrt{2S}\left( 1-\frac{\hat{n}}{2S}\right),\quad\!\!\!\hat{S}^{z}=S-\hat{n}.
\label{eq:DM_dual}
\end{equation}
For Hamiltonians depending only on combinations such as
$\hat S^+\hat S^-$, $\hat S^-\hat S^+$, and $\hat S^z$,
the representations \eqref{eq:DM_0} and \eqref{eq:DM_dual} lead to the same explicit bosonic form, since
these combinations are fixed by the algebra. More general terms,
such as those involving $\hat S^x$ or unbalanced products of ladder
operators, generally lead to different bosonic expressions.
Nevertheless, on the physical subspace the two exact representations
are related by a similarity transformation and therefore have the
same spectrum. Differences may arise after truncation, approximation,
or extension beyond the physical subspace.
\par For antiferromagnetic systems, one uses the antiferromagnetic DM mapping based on a decomposition of the lattice into two sublattices, $A$ and $B$~\cite{ueda2010supersolid}:
\begin{align}
\label{eq:ADM}
\hat S_l^{+} &= \sqrt{2S}\,\hat a_l, &
\hat S_l^{-} &= \sqrt{2S}\,\hat a_l^{\dagger} \left(1-\frac{\hat a_l^{\dagger}\hat a_l}{2S}\right), \quad l\in A, \notag \\
\hat S_m^{+} &= \sqrt{2S}\,\hat b_m^{\dagger}, &
\hat S_m^{-} &= \sqrt{2S} \left(1-\frac{\hat b_m^{\dagger}\hat b_m}{2S}\right)\hat b_m, \quad m\in B, \\
\hat S_l^{z} &= S-\hat a_l^{\dagger}\hat a_l, &
\hat S_m^{z} &= -S+\hat b_m^{\dagger}\hat b_m. \notag
\end{align}
In this case, independent bosonic operators are assigned to the spin operators on sublattices $A$ and $B$; they satisfy $[\hat a_l,\hat a_{l'}^\dagger]=\delta_{ll'}$ and $[\hat b_m,\hat b_{m'}^\dagger]=\delta_{mm'}$. The choice of vacuum is fixed by the quantization direction chosen for $S_m^z$, that is, by expanding about the values $\pm S$. This reflects the fact that, in the classical limit, the spins on sublattices $A$ and $B$ are oriented along the $+z$ and $-z$ axes, respectively. The antiferromagnetic version of the HP transformation was derived in Refs.~\cite{kubo1952spin, oguchi1960theory}.
\par
Finally, the interpolation between \eqref{eq:HP_0} and \eqref{eq:DM_0} is described by the Dembinski transformation~\cite{dembinski1964dyson}, also known as the Cooke--Loly transformation~\cite{loly1971heisenberg, cooke1970heisenberg, garbaczewski1978method}:
\begin{eqnarray}
\hat{S}^{+} &=& \sqrt{2S}\left( 1-\frac{\hat{n}}{2S}\right) ^{\alpha }
\hat{a},\quad \!\hat{S}^{-} = \hat{a}^{\dagger }\sqrt{2S}
\left( 1-\frac{\hat{n}}{2S}\right)^{1 - \alpha } \nonumber\\
\hat{S}^{z}&=&S-\hat{n},\quad 0\leq \alpha \leq 1
\label{eq:Dem_0}
\end{eqnarray}
For $\alpha=0$, we recover \eqref{eq:DM_dual}; for $\alpha=1/2$, \eqref{eq:HP_0}; and for $\alpha=1$, \eqref{eq:DM_0}. An alternative parametrization obtained through the replacement $\alpha\leftrightarrow 1-\alpha$ is also used in the literature; in this convention, the parameter $\alpha$ covers only one of the two variants of the DM transformation (see, e.g., Ref.~\cite{rudoy_bogoliubov-tyablikov_2011}).
\par It should be noted that Dembinski~\cite{dembinski1964dyson} did not give the explicit expression \eqref{eq:Dem_0}. Instead, his analysis of the choice of metric and the identification of the physical subspace singled out the three points \eqref{eq:HP_0}, \eqref{eq:DM_0}, and \eqref{eq:DM_dual}. The continuous power-law interpolation between them was later proposed heuristically by Loly~\cite{loly1971heisenberg} and is therefore often referred to as the Cooke--Loly transformation~\cite{garbaczewski1978method}. In particular, this interpolation has been used in Green-function-based diagrammatic approaches to motivate the choice of the non-Hermitian Dyson transformation~\cite{loly1971heisenberg, cooke1970heisenberg}.
\subsection{Constraints and Sources of Error in Single-Mode Transformations}

In the transformations \eqref{eq:HP_0}, \eqref{eq:DM_0}, and \eqref{eq:Dem_0}, the spin operators on the left-hand side act in the finite-dimensional Hilbert space $\dim\mathcal H_S=2S+1$, whereas the bosonic operators on the right-hand side act in the infinite-dimensional Fock space $\mathcal F$. This mismatch is resolved by observing that these mappings are in fact realized on the physical subspace $\mathcal H_{\mathrm{phys}}\subset\mathcal F$, with $\dim\mathcal H_{\mathrm{phys}}=2S+1$, selected by the constraint $0\le n\le 2S$. This establishes the one-to-one correspondence $\ket{S,m}\leftrightarrow\ket{n=S-m}$, where $\hat S^z|S,m\rangle=m|S,m\rangle$.
\par
The presence of this constraint is explicit in the HP transformation \eqref{eq:HP_0}, where the radicand $1-\frac{\hat n}{2S}$ must be nonnegative. For the DM transformation \eqref{eq:DM_0}, and for the more general transformation \eqref{eq:Dem_0}, the need for the constraint is less apparent. It must nevertheless be imposed in all calculations, since violating it leads to incorrect results, particularly in numerical implementations\footnote{States with $n>2S$ are commonly called unphysical or spurious states~\cite{vogl_resummation_2020, geyer2004non}.}. Similar problems arise if the operator $(1-\frac{\hat n}{2S})^{\frac12}$ in the HP transformation is expanded in a series without respecting this constraint. Other errors associated with expansions of operator functions are discussed in Sec.~\ref{subsec:exactness_vs_approx_ch3}.
\par A strictly finite-dimensional description based on the constraint $0\le n\le 2S$ is conceptually correct but often technically inconvenient. Practical calculations usually consider the large-spin, low-magnon-number regime
$\langle \hat n\rangle \ll 2S$,
which corresponds to the low-temperature spin-wave approximation. In this regime, the total weight of states with $n>2S$ is suppressed,
and within the usual perturbative spin-wave expansion their contributions
can be neglected at the orders normally retained~\cite{oguchi1960theory}. In particular, the contribution of states with $n>2S$ becomes negligible in the spin-wave limit~\cite{ma2011quantum}.
\par Spin--boson transformations are used far beyond magnon theory. In particular, they appear in models of collective light--matter interaction, such as the Dicke model and its generalizations, in collective many-body spin models of the Lipkin--Meshkov--Glick type, and in nuclear physics~\cite{hepp1973superradiant,wang1973phase,emary2003chaos,dimer2007proposed,dusuel2004finite,morrison2008dynamical,klein1991boson,ma2011quantum,kurucz2010multilevel}. These transformations also play an important role in quantum computing and hybrid quantum systems~\cite{liu2026hybrid, descamps2024superselection}, where they provide a bridge between continuous-variable and discrete-variable systems. A detailed discussion of these applications lies beyond the scope of the present work and can be found in Ref.~\cite{orlov2025discrete}.
\par In these settings, however, the spin-wave approximation is often inapplicable because the relevant spin and boson occupation numbers become comparable and truncation errors become significant. It is then essential to distinguish the exact realization on $\mathcal H_{\mathrm{phys}}$ from approximations outside the physical subspace. This requires the constraint $0\le n\le 2S$ to be imposed explicitly and brings us back to a formulation based on identifying the physical subspace. Subject to this constraint, Eqs.~\eqref{eq:HP_0}, \eqref{eq:DM_0}, and \eqref{eq:Dem_0} should be understood as exact: they define an isomorphism of $\mathfrak{su}(2)$ representations between
the finite-dimensional spin space $\mathcal H_S$ and the corresponding physical bosonic
subspace $\mathcal H_{\mathrm{phys}}$. The extension of the bosonic operators from the physical subspace to the full Fock space is, however, not unique.
\par A key feature of the DM transformation, and more generally of the Dembinski transformation, is the non-Hermiticity of the representation: $\hat S^+\neq(\hat S^-)^\dagger$. Consequently, matrix elements evaluated in the bosonic representation do not coincide with the canonical matrix elements of the spin operators. In the spin basis $|S,m\rangle$, one has
\begin{align}
\hat S^\pm|S,m\rangle=\sqrt{(S\mp m)(S\pm m+1)}|S,m\pm 1\rangle.
\end{align}
Upon identifying $|n\rangle\leftrightarrow|S,m=S-n\rangle$, the physically correct matrix elements in the $n$ basis must be
\begin{align} \langle n-1|S^+|n\rangle=\sqrt{n}\sqrt{2S-n+1}, \nonumber \\ \langle n+1|S^-|n\rangle=\sqrt{n+1}\sqrt{2S-n}. \end{align}
These are precisely the matrix elements produced by the HP transformation. By contrast, with the standard Fock-space inner product, the DM transformation yields
\begin{align}
\langle n-1| S^+|n\rangle=\frac{2S-n+1}{\sqrt{2S}}\sqrt{n}, \nonumber\\ \langle n+1| S^-|n\rangle=\sqrt{2S}\sqrt{n+1},
\end{align}
which do not coincide with the canonical spin expressions. This discrepancy arises because, on $\mathcal H_{\mathrm{phys}}$, the standard inner product must be replaced by a metric-modified one. Specifically, one introduces a positive metric operator $\hat\eta$ such that Hermitian conjugacy is restored with respect to the inner product $\bra{\phi}\ket{\psi}_\eta = \mel{\phi}{\hat\eta}{\psi}$. In this sense, the DM transformation defines a quasi-Hermitian rather than an ordinary Hermitian representation~\cite{mostafazadeh2001pseudo, mostafazadeh2010pseudo}. Incorporating the appropriate metric restores the physically correct matrix elements and ensures that the Hamiltonian has a real spectrum.

The quasi-Hermiticity of the DM transformation is well known in the mathematical literature and must be taken into account in accurate numerical calculations~\cite{mostafazadeh2010pseudo, scholtz1992quasi}\footnote{Quasi-Hermiticity and metric operators were already used, in effect, in the original works of Dyson and Dembinski~\cite{dyson1956thermodynamic, dyson1956general, dembinski1964dyson} and were also discussed in Ref.~\cite{mills1966problems}. The corresponding conceptual framework and modern terminology, however, were developed considerably later.}. In particular, the transformations \eqref{eq:DM_0} and \eqref{eq:DM_dual} are related by a similarity transformation on the
physical subspace but differ in their extensions beyond the physical subspace. Consequently, under truncation, approximation, or numerical implementation, they may produce different errors and approximate Hamiltonians with different structures. Nevertheless, in many physical problems the effect of the nontrivial metric is asymptotically small. For a single spin $S=\tfrac12$, the metric operator is the identity, while in the spin-wave limit its diagonal elements satisfy $\eta_n=1-n(n-1)/(4S)+O(n^4/S^2)$. This asymptotic approach to the identity explains why metric corrections are usually neglected in standard calculations of magnon theory. Beyond the spin-wave approximation, however, they may become important.\footnote{This also explains the alternative form of the DM transformation encountered mainly in the mathematical literature, $\hat S^+=(2S-\hat n)\hat a$,
$\hat S^-=\hat a^\dagger$~\cite{ivanov2008dyson}. This form is algebraically valid and preserves the commutation relations
because, after identifying the spin space with the physical bosonic
subspace, it is related to the standard Dyson--Maleev realization by an
invertible change of basis within the same finite-dimensional space. It is more compact and convenient for algebraic calculations; however, like the ``physical'' version, it yields matrix elements that differ from the standard ones and are restored only after a positive metric is introduced. From a physical perspective, its main drawback is that it does not simplify in the spin-wave limit and the corresponding metric does not approach the identity, making numerical implementation more difficult.}. A detailed investigation of the quasi-Hermitian nature of these transformations lies outside the scope of the present work. Below, metric-related questions are addressed only insofar as necessary for a correct interpretation of non-Hermitian spin--boson representations.

\selectlanguage{english}
\section{Two-Mode Transformations}
\label{sec:two_mode}
\par Another widely used spin--boson representation is the Jordan--Schwinger (JS) transformation, also known as the Schwinger boson representation~\cite{Jordan1935,Schwinger1952}, in which the $\mathfrak{su}(2)$ algebra is realized using two independent bosonic modes.

Let $\hat a$ and $\hat b$ denote two independent bosonic modes satisfying $[\hat a,\hat a^\dagger]=[\hat b,\hat b^\dagger]=1$, with all other commutators vanishing. Each spin projection is associated with a distribution of the boson number between the two modes: the $\hat a$ mode accounts for an increase in the spin projection, whereas the $\hat b$ mode accounts for a decrease. In particular, the configuration $n_b=0$, $n_a=2S$ is chosen as the reference state and corresponds to the state of maximal spin projection, $m=+S$. The spin--boson transformation then reads
\begin{equation}
\hat{S}^{+}=\hat{a}^{\dagger }\hat{b},\quad \hat{S}^{-}=\hat{b}^{\dagger }%
\hat{a},\quad \hat{S}^{z}=\frac{1}{2}\left( \hat{a}^{\dagger }\hat{a}-\hat{b}%
^{\dagger }\hat{b}\right)
\label{eq:SlB_0}
\end{equation}
subject to the simultaneous conditions
\begin{equation}
n_a+n_b=2S,
\qquad
m=\frac{n_a-n_b}{2},
\label{eq:SlB_cons}
\end{equation}
where $\hat n_a=\hat a^\dagger\hat a$ and $\hat n_b=\hat b^\dagger\hat b$ are the number operators, while $n_a$ and $n_b$ are their eigenvalues. The constraint \eqref{eq:SlB_cons} selects a physical sector of dimension $2S+1$. The basis states are identified according to
\begin{align}
\ket{S,m}\;\longleftrightarrow\;\ket{n_a=S+m,\;n_b=S-m}.
\end{align}
In magnetism, the Schwinger-boson representation is most commonly applied to the case $S=\tfrac12$; formally, however, it is not restricted to this value and can be used for arbitrary spin~\cite{dubus2024bosons}.

\par The JS and HP transformations are related. Within the physical sector,
one of the two modes can be eliminated using \eqref{eq:SlB_cons}, for example
by writing
$\hat a^\dagger\hat a=2S-\hat b^\dagger\hat b$.
For a simple heuristic illustration, after fixing the phase convention, one
may represent the eliminated mode by
$\sqrt{2S-\hat b^\dagger\hat b}$.
This replacement is understood only on the constrained sector and does not
imply the operator identity $\hat a=\hat a^\dagger$ in the full two-mode
Fock space. One then obtains
\begin{eqnarray}
\hat{S}^{+}
&=&
\sqrt{2S-\hat{b}^{\dagger }\hat{b}}\,\hat b,
\quad
\hat{S}^{-}
=
\hat b^{\dagger}\sqrt{2S-\hat{b}^{\dagger }\hat{b}},
\quad
\hat{S}^{z}
=
S-\hat b^{\dagger}\hat b,
\end{eqnarray}
which, up to a relabeling of the operators, reproduces the HP expression \eqref{eq:HP_0}.
A rigorous derivation of this reduction can be found in
Ref.~\cite{orlov2025discrete}.

\par In the theory of strongly correlated systems, a related family of methods is collectively known as \emph{slave-particle} representations. Their basic idea is to express the original local degree of freedom, such as an electron at a lattice site, in terms of a set of auxiliary bosonic and/or fermionic operators and then impose a constraint that selects the physical subspace. Depending on the particular factorization, the charge and spin degrees of freedom can be described by different combinations of auxiliary bosons and fermions. Well-known examples include \emph{slave-boson} approaches~\cite{Auerbach1994,wolfle_slave_1995,Izyumov_1995} and representations employing Schwinger bosons in the $t$--$J$ and Hubbard models~\cite{Izyumov_1995,coleman_1984}. Their close connection to spin--boson transformations lies not only in the use of auxiliary bosonic degrees of freedom but also in a common construction principle: the original state space is enlarged and the physical subspace is subsequently selected by imposing a constraint.

\selectlanguage{english}
\section{The \texorpdfstring{$f$}{f}-Deformed Oscillator Formalism}
\label{sec:fosc}
Consider a canonical bosonic mode specified by the creation operator $\hat a^\dagger$ and the annihilation operator $\hat a$, satisfying the canonical commutation relation $[\hat a,\hat a^\dagger]=1$, and by the number operator $\hat n=\hat a^\dagger\hat a$. A broad class of nonlinear (``$f$-deformed'') oscillator models is obtained by modifying the canonical bosonic operators with a function of the number operator~\cite{man1997f,man1998nonlinearity}:
\begin{equation}
\hat A \equiv \hat a f(\hat n),\qquad
\hat A^\dagger \equiv f(\hat n) \hat a^\dagger,
\label{eq:f_def_main}
\end{equation}
where $f:\mathbb N_0\to\mathbb R$ is a numerical function, and the operator $f(\hat n)$ is defined spectrally in the Fock basis ${\ket{n}}$ by $f(\hat n)\ket{n}=f(n)\ket{n}$. Without loss of generality, $f$ may be taken to be real and nonnegative, so that $f(\hat n)=f^\dagger(\hat n)$~\cite{man1997f}.

It is convenient to encode the deformation in terms of the \textit{structure function} $\Phi(\hat n)\equiv\hat A^\dagger\hat A$:
\begin{equation}
\Phi(\hat n) =\hat n\,f^2(\hat n),
\quad
\Phi(n)=n\,f^2(n),\quad \Phi(0)=0,
\label{eq:phi_main}
\end{equation}
which determines the matrix elements in the Fock basis:
\begin{align}
\hat A\ket{n} = 
\sqrt{\Phi(n)}\,\ket{n-1}.
\label{eq:action_main}
\end{align}
Similarly, $\hat A^\dagger\ket{n}=\sqrt{\Phi(n+1)}\,\ket{n+1}$. Using the shift identity $\hat a\,F(\hat n)=F(\hat n+1)\,\hat a$ (see Appendix~\ref{app:shift_domains}), we obtain the commutator
\begin{equation}
[\hat A,\hat A^\dagger]=\Phi(\hat n+1)-\Phi(\hat n),
\label{eq:comm_main}
\end{equation}
where different choices of $\Phi(\hat n)$ encode different nonlinear oscillators. The commutators close into a Lie algebra only for special structure functions.

\subsection{Single-Mode Spin--Boson Transformations}
\label{sec:one_mode_diff}

We seek a realization of the $\mathfrak{su}(2)$ generators ${\hat S^+,\hat S^-,\hat S^z}$ satisfying the commutation relations \eqref{eq:su2_comm_ch3_1} and \eqref{eq:su2_comm_ch3} in the finite-dimensional Hilbert space $\mathcal H_S$, with $\dim\mathcal H_S=2S+1$, corresponding to spin $S$. The Fock-space vacuum is chosen as the reference state and identified with the highest-weight state $m=+S$. This fixes the realization of $\hat S^z$ as
$\hat S^z = S-\hat n$,
and imposes, on the physical subspace, the boundary conditions
\begin{equation}
\hat S^+|n = 0\rangle=0,\qquad \hat S^-|n = 2S\rangle=0.
\label{eq:hw_lw_conditions}
\end{equation}
Following the $f$-oscillator formalism \eqref{eq:f_def_main}, consider the single-mode ansatz
\begin{equation}
\hat S^+ = \hat a f_1(\hat n),\quad
\hat S^- = f_2(\hat n)\hat a^\dagger,\quad
\hat S^z = S-\hat n,
\label{eq:universal_one_mode_ansatz_ch3}
\end{equation}
where $f_1(\hat n)$ and $f_2(\hat n)$ are generally distinct operator functions specified by numerical functions $f_{1,2}:\mathbb N_0\to\mathbb C$. Unlike in the standard Hermitian $f$-oscillator formalism, we do not assume that $f_1=f_2$ and therefore do not impose the Hermiticity condition $f(\hat n)=f^\dagger(\hat n)$. This formulation encompasses both Hermitian and non-Hermitian realizations because $\hat S^+$ and $\hat S^-$ need not be related by Hermitian conjugation.
\par Irrespective of the choice of $f_{1,2}(\hat n)$, the relations \eqref{eq:su2_comm_ch3_1} are satisfied. Thus, the only nontrivial $\mathfrak{su}(2)$ condition is \eqref{eq:su2_comm_ch3}.
Using the identities in Appendix~\ref{app:shift_domains}, one then readily obtains
\begin{align}
\hat S^- \hat S^+ &=
f_2(\hat n)\,\hat a^\dagger \hat a\, f_1(\hat n)
= \hat n\,G(\hat n),
\label{eq:SmSp_ch3}\\[2mm]
\hat S^+ \hat S^- &=
\hat a\, f_1(\hat n) f_2(\hat n)\,\hat a^\dagger
= (\hat n+1)\,G(\hat n+1),
\label{eq:SpSm_ch3}
\end{align}
where
\begin{align}\label{1422}G(\hat n) \equiv f_1(\hat n)\,f_2(\hat n).\end{align}
Consequently, the commutator condition $[\hat S^+,\hat S^-]=2(S-\hat n)$ reduces to the operator finite-difference equation
\begin{equation}
(\hat n+1)G(\hat n+1)-\hat nG(\hat n)=2(S-\hat n).
\label{eq:diff_eq_operator_ch3}
\end{equation}
In the Fock basis ${|n\rangle}$, this equation is equivalent to a first-order recurrence relation on $\mathbb N_0$. Its solution follows from standard methods for finite-difference equations and is given in Appendix~\ref{app:diff_eq_solution}. The recurrence uniquely determines all values of $G(n)$ that enter the operator products on $\mathcal H_{\mathrm{phys}}$:
\begin{equation}
G(n)=2S-n+1,\qquad n=1,\ldots,2S+1
\label{eq:G_solution_ch3}
\end{equation}
or, equivalently,
\begin{equation}
\hat S^- \hat S^+ = \hat n\,(2S-\hat n+1),\quad
\hat S^+ \hat S^- = (\hat n+1)\,(2S-\hat n).
\label{eq:products_fixed_ch3}
\end{equation}
The value $G(0)$ is not fixed by the recurrence, since it always appears multiplied by $\hat n$. Throughout this work, we adopt the convenient convention $G(0)=2S+1$, corresponding to the continuation of the same expression to $n=0$.

Equation~\eqref{eq:diff_eq_operator_ch3} plays a central role in what follows. Because the $\mathfrak{su}(2)$ algebra fixes only the product \eqref{1422}, there remains a freedom in the choice of representation arising from the nonuniqueness of the factorization of $G$ into the factors $f_1$ and $f_2$. This observation is essential for constructing different realizations. We show below that the known Holstein--Primakoff, Dyson--Maleev, and Jordan--Schwinger transformations arise naturally from this condition. We also obtain new transformations with their own constraints and advantages. Of particular interest is that this approach can be generalized naturally both to multimode representations and to other algebras for which bosonic and non-Hermitian realizations have not yet been constructed. The importance of searching for bosonic representations of other algebras is also noted in Ref.~\cite{rudoy_bogoliubov-tyablikov_2011}.

\subsection{\texorpdfstring{The \(\alpha\)-Family}{The alpha-Family}}
\label{subsec:alpha_family_ch3}

A convenient class of solutions to Eq.~\eqref{eq:diff_eq_operator_ch3} is obtained through a factorization based on a \textit{power-law} splitting, or equivalently through the single-mode $\alpha$-family. First, observe that the solution $G(n)=2S-(n-1)$ can be rewritten as
\begin{equation}\label{sqrt_2s_G}
2S\,\frac{G(n)}{2S}
=
2S\left(1-\frac{n-1}{2S}\right).
\end{equation}
This form is particularly convenient for constructing the so-called ``physical'' version of the transformations because it naturally introduces the small parameter $(n-1)/(2S)$ used in the spin-wave approximation. Without this normalization, one obtains a more formal ``mathematical'' version of the same transformations. Equation~\eqref{1422} gives
\begin{equation}
\!\!\!\!f_{1\alpha}(\hat n)=\sqrt{2S}\left(\frac{G(\hat n)}{2S}\right)^{\alpha}\!\!\!\!,\quad \!\!\!\!
f_{2\alpha}(\hat n)=\sqrt{2S}\left(\frac{G(\hat n)}{2S}\right)^{1-\alpha}\!\!\!\!\!\!\!\!,
\label{eq:alpha_split_ch3}
\end{equation}
where $\alpha\in\mathbb R$, and the normalization factor $2S$ in \eqref{sqrt_2s_G} is distributed symmetrically between the two factors. This factorization generates a continuous $\alpha$-family of realizations on $\mathcal H_{\mathrm{phys}}$:
\begin{eqnarray}
\label{eq:alpha_family_Splus_ch3}
\hat S_{\alpha}^+\!\! &=&\!\! \sqrt{2S}\left(1-\frac{\hat n}{2S}\right)^{\alpha}\hat a
\;\equiv\;
\sqrt{2S}\,\hat a\,\left(1-\frac{\hat n-1}{2S}\right)^{\alpha},
\qquad\\
\hat S_{\alpha}^-\!\! &=&\!\! \sqrt{2S}\,\hat a^\dagger\left(1-\frac{\hat n}{2S}\right)^{1-\alpha}
\!\!\!\!\!\! \equiv\;
\sqrt{2S}\left(1-\frac{\hat n-1}{2S}\right)^{1-\alpha}\!\!\hat a^\dagger.
\label{eq:alpha_family_Sminus_ch3}
\end{eqnarray}
The equivalence of the different forms follows from the shift identities given in Appendix~\ref{app:shift_domains}. The family \eqref{eq:alpha_family_Splus_ch3}--\eqref{eq:alpha_family_Sminus_ch3} contains several standard spin--boson transformations as special cases. For $\alpha=\tfrac12$, it yields the Hermitian Holstein--Primakoff transformation \eqref{eq:HP_0}; for $\alpha=1$, the non-Hermitian Dyson--Maleev transformation \eqref{eq:DM_0}; and for $\alpha=0$, the transformation dual to DM, \eqref{eq:DM_dual}. Thus, $\alpha$ provides a continuous interpolation between a Hermitian realization and two mutually dual non-Hermitian realizations of $\mathfrak{su}(2)$. Consequently, the HP and DM transformations, together with the intermediate realizations, should be viewed not as isolated constructions but as distinguished points in a single continuum of solutions arising from different factorizations of the same algebraically fixed object on the physical subspace.

The novelty of the proposed approach lies in its rigorous algebraic derivation: the universal difference equation \eqref{eq:diff_eq_operator_ch3} singles out the invariant function \eqref{1422} shared by all such realizations, including the power-law $\alpha$-family. Different bosonic representations therefore arise not from independent ansatz constructions but as different factorizations of the same algebraically determined structure. In this formulation, the algebra itself is the starting point rather than a choice of metric or a Hermiticity condition; the latter may be introduced subsequently, depending on the representation required for a particular physical or computational problem. It is also clear that, in general, the parameter $\alpha\in\mathbb R$ in "physical" form need not be restricted to the interval $[0,1]$. For arbitrary real $\alpha$, these expressions are understood through their action on the physical ladder. The shifted deformation functions involving $1-(\hat n-1)/(2S)$ are manifestly finite on the physical spectrum $n=0,\ldots,2S$; possible endpoint singularities in the composite operator expressions do not enter nonvanishing physical matrix elements, since the highest- and lowest-weight actions are fixed separately. The shifted and unshifted forms remain equivalent on their common domain of definition.

The power-law $\alpha$-splitting \eqref{eq:alpha_split_ch3} does not exhaust the full factorization freedom of the solution \eqref{1422}. A more general class of factorizations is obtained by choosing an arbitrary invertible function $u(\hat n)$ on $\mathcal H_{\mathrm{phys}}$. For any $\alpha\in\mathbb R$, one may set
\begin{equation}
f_1(\hat n)=u(\hat n)\,G(\hat n)^\alpha,
\quad
f_2(\hat n)=u(\hat n)^{-1}G(\hat n)^{1-\alpha},
\label{eq:general_u_factorization}
\end{equation}
which is, after absorbing a constant normalization into $u(\hat n)$, equivalent to
\begin{equation}
f_1(\hat n)=u(\hat n)\,f_{1,\alpha}(\hat n),
\qquad
f_2(\hat n)=u(\hat n)^{-1}f_{2,\alpha}(\hat n),
\label{eq:general_u_factorization_alpha}
\end{equation}
where the functions $f_{1,\alpha}(\hat n)$ and $f_{2,\alpha}(\hat n)$ are defined by the power-law family \eqref{eq:alpha_split_ch3}. Thus, the $\alpha$-family represents only a special case of the general factorization freedom, whereas the choice of $u(\hat n)$ introduces an additional ``gauge'' degree of freedom not fixed by the algebra.
\par This freedom may be important in applications because different choices of $u(\hat n)$ redistribute the nonlinearity between $\hat S^+$ and $\hat S^-$ in different ways. Consequently, they change not only the explicit form of the corresponding representation but also the structure of the metric operator in the quasi-Hermitian formulation and the numerical behavior of the model after truncations, expansions, and other approximations. In practice, this freedom can be used to select the representation best suited to a particular problem---for example, to simplify the Hamiltonian as much as possible or to minimize the energy error, commutator defect, deviation of the metric from the identity operator, or other characteristics of an approximate realization.

For example, admissible classes of such factorizations include rational and exponential functions of the form
\begin{equation}
u(n)=\frac{P(n)}{Q(n)},
\qquad
u(n)=e^{\phi(n)},
\label{eq:u_examples}
\end{equation}
provided that $u(n)\neq0$ for all $n=0,1,\dots,2S$. These examples demonstrate that the freedom in choosing $u(n)$ is considerably broader than the power-law $\alpha$-family and cannot be reduced to a simple constant renormalization of known realizations. In the present work, these more general factorizations are considered primarily as an indication of the underlying structure of the problem and of possible directions for further generalization. A detailed analysis, including the question of which choices of $u(n)$ are most convenient in specific physical or computational applications, lies beyond the scope of the present study.
\subsection{Approximations and Exactness on the Physical Subspace}
\label{subsec:exactness_vs_approx_ch3}
The expressions \eqref{eq:alpha_family_Splus_ch3}--\eqref{eq:alpha_family_Sminus_ch3} are exact on the finite subspace $\mathcal H_{\mathrm{phys}}$ when the appropriate constraint is imposed. Inaccuracies arise only when operator functions such as $\sqrt{2S-\hat n}$ are replaced by approximate expansions, for example Taylor series, which are then truncated. In general, such truncations violate the algebraic relations already at the operator level.
\par A Taylor expansion does not exploit the discrete nature of the spectrum of $\hat n$ and is therefore not the most natural tool for approximating functions of the occupation number. Newton series provide a more suitable alternative. The idea of using such expansions for spin--boson representations was first proposed in Ref.~\cite{kubo1952spin}, although that work subsequently employed the large-spin approximation and an ordinary Taylor expansion. The approach was later rediscovered independently in Refs.~\cite{vogl_resummation_2020, konig2021newtonseriesexpansionbosonic}. Newton series not only substantially improve computational accuracy relative to conventional truncated expansions but also conform naturally to the discrete structure of the spectrum of $\hat n$. Moreover, in the operator formulation, they automatically yield normally ordered expressions in the creation and annihilation operators. Since the functions $f(\hat n)$, $\Phi(\hat n)$, and $G(\hat n)$ in the $f$-deformed oscillator formalism are functions of a discrete operator, Newton series constitute their natural computational complement. Whereas the $f$-oscillator approach determines which algebra is realized and how it is encoded in operator functions, the Newton-series approach provides an efficient tool for expanding and approximating these functions. On the finite physical spectrum, the Newton expansion can be made exact. Its polynomial continuation to $n>2S$, however, is not constrained by the finite-dimensional spin representation and need not reproduce the same operator relations on the full Fock space. Details, including an application to the $\alpha$-family and the relation between finite-block exactness and global polynomial realizations, are given in Appendix~\ref{app:newton_series}.

\subsection{Filtered Off-Block Extensions}
\label{subsec:filters}
\par So far, we have considered transformations that are exact on the physical block and that, outside this block, either vanish automatically or can be interpreted correctly only after the appropriate constraint is imposed. As noted above, the DM transformation does not have this property automatically; it acquires it only after the constraint $0\le n\le2S$ is imposed explicitly. Likewise, approximations and expansions of the square root $\sqrt{2S-\hat n}$ in the HP transformation may generate spurious states outside the physical block. Thus, the algebraic relations and the difference equation \eqref{eq:diff_eq_operator_ch3} fix the invariant product only on the physical block, whereas the extension of the operators to the full Fock space $\mathcal F$ outside this block is not unique. This introduces an additional freedom that may naturally be interpreted as the freedom of \emph{off-block extension}. In other words, the same exact realization of $\mathfrak{su}(2)$ on the physical block can be extended to the full space $\mathcal F$ in different ways. From a practical standpoint, this motivates the construction of off-block extensions that preserve equivalence on the physical block while facilitating the analytical or numerical suppression of unphysical states.

We say that two sets of operators,
$\{\hat S^i\}$ and $\{\widetilde S^i\}$,
$i\in\{+,-,z\}$,
are \emph{equivalent on the physical subspace} if
\begin{equation}
\Pi_{2S}\,\widetilde S^i\,\Pi_{2S}
=
\Pi_{2S}\,\hat S^i\,\Pi_{2S},
\qquad
i\in\{+,-,z\},
\label{eq:block_equivalence}
\end{equation}
where $\Pi_{2S}$ is the projector onto the physical subspace. To construct such equivalent off-block extensions, introduce a smooth function $W(\hat n)$, which we call a \emph{filter}. It is assumed to leave the physical subspace unchanged while rapidly suppressing unphysical states. It is sufficient to require
\begin{equation}
W(n)=1,
\qquad
n=0,1,\dots,2S.
\label{eq:W_block_identity}
\end{equation}
The \emph{filtered} operators can then be defined as
\begin{equation}
\widetilde S^\pm
\equiv
W(\hat n)\,\hat S^\pm\,W(\hat n),
\qquad
\widetilde S^z
\equiv
W(\hat n)\,\hat S^z\,W(\hat n).
\label{eq:soft_filtered_ops}
\end{equation}
By construction, these operators are equivalent to the original ones on the physical subspace but may differ substantially outside it. One convenient example is the following family of filters:
\begin{equation}
W_\lambda(\hat n)
=
\exp\!\left[
-\lambda
\left(
\prod_{k=0}^{2S}(\hat n-k)
\right)^2
\right],
\qquad
\lambda>0.
\label{eq:analytic_filter_one_mode}
\end{equation}
Condition~\eqref{eq:W_block_identity} is satisfied automatically because, for $n=0,1,\dots,2S$, the product $\prod_{k=0}^{2S}(n-k)$ vanishes and therefore $W_\lambda(n)=1$. Outside the physical subspace, for $n>2S$, the magnitude of the product grows rapidly, leading to an exponential suppression of $W_\lambda(n)$.

\selectlanguage{english}
\section{Two-Mode Spin--Boson Transformations}
\label{sec:two_mode_js}
\par Consider the ansatz
\begin{equation}
\hat S^+ = \hat a_1^\dagger \hat a_2 f(\hat n_1,\hat n_2),\quad
\hat S^- = \hat a_2^\dagger \hat a_1 g(\hat n_1,\hat n_2),
\label{eq:JS_right_local}
\end{equation}
where $f$ and $g$ are, in general, distinct operator-valued functions defined spectrally on $\mathbb N_0^2$. Allowing both functions to depend on the two number operators $\hat n_1$ and $\hat n_2$ ensures maximal generality within the class of realizations under consideration. This ansatz is right-ordered, since $f$ and $g$ are placed to the right of the corresponding combinations of creation and annihilation operators. Appendix~\ref{app:JS_local_forms} shows that the left- and center-ordered forms are equivalent to the right-ordered one.
\par
In this section, we consider the class of representations that conserve the total boson number
$\hat N=\hat n_1+\hat n_2$, so that
\begin{align}\label{777}
    [\hat N,\hat S^\pm]=0
\end{align}
and whose diagonal generator is
\begin{eqnarray}\label{672}
\hat S^z=\frac{\hat n_1-\hat n_2}{2}.
\end{eqnarray}
This class includes both the standard Jordan--Schwinger transformation and its deformed generalizations. It should be emphasized, however, that condition \eqref{777} is not part of the definition of the $\mathfrak{su}(2)$ algebra itself; rather, it is an additional restriction imposed on the class of its bosonic realizations. Relaxing this condition leads to a broader family of transformations beyond the standard Jordan--Schwinger scheme.
\par
It is readily verified that the ansatz \eqref{eq:JS_right_local} satisfies both \eqref{eq:su2_comm_ch3_1} and \eqref{777}, irrespective of the choice of $f$ and $g$. Thus, the only nontrivial condition imposed by the $\mathfrak{su}(2)$ algebra is \eqref{eq:su2_comm_ch3}.
In the number basis, this condition is equivalent, for all $(n_1,n_2)\in\mathbb N_0^2$, to the discrete equation (see Appendix~\ref{app:JS_lattice_constraint_derivation})
\begin{equation}
n_1(n_2+1)\,H(n_1-1,n_2+1)-n_2(n_1+1)\,H(n_1,n_2)
=n_1-n_2,
\label{eq:lattice_constraint}
\end{equation}
At the boundary, terms multiplied by vanishing occupation-number factors are omitted. Here we have introduced the shifted product
\begin{equation}
H(n_1,n_2)\equiv f(n_1,n_2)\,g(n_1+1,n_2-1),\quad (n_2\ge 1)
\label{eq:H_def}
\end{equation}
Now fix the standard physical sector $\mathcal H_{2S}$ of the Jordan--Schwinger transformation, defined by the constraint $\hat N=2S$. On this sector, it is convenient to use the notation
\begin{equation}
\ket{n}_S \equiv \ket{2S-n,n}\,:\, n=0,1,\dots,2S
\label{eq:JS_block_basis}
\end{equation}
and to define
\begin{equation}
A_n \equiv f(2S-n-1,n+1)g(2S-n,n)\equiv f_{n+1}g_{n},
\label{eq:An_def}
\end{equation}
where $n=0,1,\dots,2S-1$. Condition \eqref{eq:su2_comm_ch3}, restricted to the block $\mathcal H_{2S}$, then reduces to the finite-difference recurrence relation
\begin{equation}
(n+1)(2S-n)\,A_n
-
n(2S-n+1)\,A_{n-1}
=
2(S-n),
\label{eq:An_recursion}
\end{equation}
where $n=0,1,\dots,2S$. At $n=0$, the second term is absent, whereas at $n=2S$, the first term is absent. This equation has the unique solution
\begin{equation}
A_n\equiv 1,\qquad n=0,1,\dots,2S-1.
\label{eq:An_one}
\end{equation}
This follows by direct induction. At $n=0$, Eq.~\eqref{eq:An_recursion} gives
$(2S)A_0=2S$, and hence $A_0=1$.
Next, if $A_{n-1}=1$, then
\begin{align}
(n+1)(2S-n)A_n
=
(n+1)(2S-n),
\end{align}
and therefore $A_n=1$. Thus, Eq.~\eqref{eq:An_one} holds for all $n=0,1,\dots,2S-1$.
\par As in the single-mode case, the $\mathfrak{su}(2)$ algebra fixes not the functions $f$ and $g$ separately, but only their invariant combination. A convenient representative of the corresponding factorization class is the power-law choice
\begin{equation}
f_n=(2S-n+1)^{\alpha-\frac12},
\qquad
g_n=(2S-n)^{\frac12-\alpha},
\label{eq:fg_alpha_two_mode}
\end{equation}
For arbitrary real $\alpha$, only the values of $f_n$ and $g_n$ that enter the action of $\hat S^\pm$ on the physical sector are required, namely $f_n$ for $n=1,\ldots,2S$ and $g_n$ for $n=0,\ldots,2S-1$. Hence their power-law bases are strictly positive for all nonvanishing physical transitions, and no endpoint singularities arise.

Upon identifying the block $\mathcal H_{2S}$ with the single-mode physical subspace
$\mathcal H_{\mathrm{phys}}$, these expressions reduce to the single-mode
$\alpha$-family discussed above in Sec.~\ref{sec:one_mode_diff}
(see Appendix~\ref{app:JS_alpha_reduction_ru}). Thus, the single-mode Holstein--Primakoff and Dyson--Maleev transformations, together with their intermediate forms, admit a natural interpretation as reductions of the same two-mode factorization freedom, restricted to a fixed Jordan--Schwinger block.

\selectlanguage{english}
\section{Two-Mode Transformations Beyond Jordan--Schwinger}
\label{sec:beyond_js}
It is natural to ask whether there exist two-mode representations that satisfy the same commutation relations \eqref{eq:su2_comm_ch3_1} and
\eqref{eq:su2_comm_ch3} but do not conserve the total boson number and therefore violate condition \eqref{777}. Constructions beyond the standard Jordan--Schwinger scheme have already been discussed in the literature~\cite{klein1991boson, tsue2015beyond}, although not within the $f$-deformation formalism. We show below that even a minimal nontrivial extension of the ansatz \eqref{eq:JS_right_local} in this direction admits an exact solution.
\par
We retain the choice \eqref{672}.
The simplest extension then takes the form
\begin{align}
\hat S^+
&=
\hat a_1^\dagger \hat a_2\, f_0(\hat n_1,\hat n_2)
+
(\hat a_1^\dagger)^2 f_{+2}(\hat n_1,\hat n_2)
+
\hat a_2^2 f_{-2}(\hat n_1,\hat n_2),
\label{eq:minimal_extended_Splus}
\\[1mm]
\hat S^-
&=
\hat a_2^\dagger \hat a_1\, g_0(\hat n_1,\hat n_2)
+
\hat a_1^2 g_{-2}(\hat n_1,\hat n_2)
+
(\hat a_2^\dagger)^2 g_{+2}(\hat n_1,\hat n_2),
\label{eq:minimal_extended_Sminus}
\end{align}
where $f_0$, $f_{\pm2}$, $g_0$, and $g_{\pm2}$ are diagonal operator-valued functions defined spectrally on $\mathbb N_0^2$. This ansatz permits only three changes in the total boson number:
\begin{align}
\Delta N=0,\qquad \Delta N=+2,\qquad \Delta N=-2.
\end{align}
Unlike the Jordan--Schwinger representation, in which the ladder operators preserve subspaces of fixed $N$, the action of $\hat S^\pm$ now mixes neighboring $N$-sectors of the same parity. Relations \eqref{eq:su2_comm_ch3_1} are satisfied identically (see Appendix~\ref{app:weight_monomials}), so that \eqref{eq:su2_comm_ch3} remains the only nontrivial condition.
\par We restrict attention to the constant subclass
\begin{equation}
f_0=a,\,
g_0=b,\,
f_{+2}=c,\,
f_{-2}=d,\,
g_{-2}=e,\,
g_{+2}=h,
\label{eq:constant_ansatz_coeffs}
\end{equation}
where $a,b,c,d,e,h\in\mathbb C$ are independent of $\hat n_1$ and $\hat n_2$. The ansatz \eqref{eq:minimal_extended_Splus}--\eqref{eq:minimal_extended_Sminus} then becomes
\begin{align}
\hat S^+
&=
a\,\hat a_1^\dagger \hat a_2
+
c\,(\hat a_1^\dagger)^2
+
d\,\hat a_2^2,
\label{eq:constant_Splus}
\\[1mm]
\hat S^-
&=
b\,\hat a_2^\dagger \hat a_1
+
e\,\hat a_1^2
+
h\,(\hat a_2^\dagger)^2,
\label{eq:constant_Sminus}
\end{align}
The operators \eqref{eq:constant_Splus}--\eqref{eq:constant_Sminus}
furnish an algebraic realization of the $\mathfrak{su}(2)$
commutation relations, i.e., satisfy
\eqref{eq:su2_comm_ch3_1} and \eqref{eq:su2_comm_ch3},
if and only if the coefficients $a,b,c,d,e,h$ obey the system of equations (see Appendix~\ref{app:weight_monomials})
\begin{equation}
ah=bc,\quad
bd=ae,\quad
ce=dh,\quad
ab-4ce=1.
\label{eq:constant_constraints}
\end{equation}
\par The system \eqref{eq:constant_constraints} defines a three-parameter
family of constant solutions.
In particular, when $a\neq 0$, $b\neq 0$, and $c\neq 0$, the parameters $a$, $b$, and $c$
may be chosen arbitrarily, while the remaining coefficients are given by
\begin{equation}
e=\frac{ab-1}{4c},
\qquad
d=\frac{ae}{b},
\qquad
h=\frac{bc}{a}.
\label{eq:constant_general_parametrization}
\end{equation}
\par Under the additional constraints $a=b$ and $c=d=e=h$, the system
\eqref{eq:constant_constraints} reduces to the single relation
$a^2-4c^2=1$.
Consequently, for any $r\in\mathbb R$, the parametrization
\begin{equation}
a=\cosh(2r),
\qquad
c=-\frac12\sinh(2r)
\label{eq:symmetric_parametrization_r}
\end{equation}
defines an exact realization of the $\mathfrak{su}(2)$ algebra:
\begin{align}
\hat S^+(r)
&=
\cosh(2r)\,\hat a_1^\dagger \hat a_2
-\frac12\sinh(2r)((\hat a_1^\dagger)^2
+\hat a_2^2),
\label{eq:symmetric_Splus_r}
\\[1mm]
\hat S^-(r)
&=
\cosh(2r)\,\hat a_2^\dagger \hat a_1
-\frac12\sinh(2r)(\hat a_1^2
+(\hat a_2^\dagger)^2),
\label{eq:symmetric_Sminus_r}
\\[1mm]
\hat S^z&=\frac{\hat n_1-\hat n_2}{2}.
\label{eq:symmetric_Sz_r}
\end{align}
\par The family \eqref{eq:symmetric_Splus_r}--\eqref{eq:symmetric_Sz_r} admits a natural interpretation as a squeezed Jordan--Schwinger representation generated by a two-mode squeezing transformation. Specifically, introduce the two-mode squeezing operator
\begin{equation}
\hat U(r)=
\exp\!\left[
r\bigl(\hat a_1^\dagger \hat a_2^\dagger-\hat a_1\hat a_2\bigr)
\right],
\label{eq:two_mode_squeeze_for_remark}
\end{equation}
\par The corresponding generators can then be written as
\begin{equation}
\hat S^\pm(r)=\hat U(r)\hat S^\pm \hat U(r)^{-1},
\label{eq:squeezed_similarity_relation}
\end{equation}
\par Thus, the symmetric solution of the constant subclass is not merely a formal special case of the system \eqref{eq:constant_constraints}; it also has a natural interpretation in terms of a Bogoliubov transformation.
\par At $r=0$, Eqs.~\eqref{eq:symmetric_Splus_r}--\eqref{eq:symmetric_Sminus_r}
reduce to the standard Jordan--Schwinger transformation. For $r\neq 0$, however, the operators $\hat S^\pm(r)$ contain the terms
$(\hat a_1^\dagger)^2$, $\hat a_2^2$, $\hat a_1^2$, $(\hat a_2^\dagger)^2$,
which change the total boson number by $\pm2$ and therefore violate condition \eqref{777}.
Nevertheless, the representation retains the conserved quantity
\begin{equation}
\hat N_r \equiv \hat U(r)\,\hat N\,\hat U(r)^{-1},
\label{eq:Nr_def}
\end{equation}
which commutes with all generators:
\begin{equation}
[\hat N_r,\hat S^\pm(r)]=0,
\qquad
[\hat N_r,\hat S^z(r)]=0.
\label{eq:Nr_commutes}
\end{equation}
Its explicit form is
\begin{equation}
\hat N_r
=
\cosh(2r)\,(\hat n_1+\hat n_2)
-
\sinh(2r)\,(\hat a_1^\dagger \hat a_2^\dagger+\hat a_1\hat a_2)
+
2\sinh^2 r.
\label{eq:Nr_explicit}
\end{equation}
Consequently, the irreducible spin blocks are mapped onto the ``squeezed'' subspaces
\begin{equation}
\mathcal H_{2S}^{(r)}
=
\hat U(r)\,\mathcal H_{2S},
\label{eq:squeezed_blocks}
\end{equation}
where $\mathcal H_{2S}$ is the standard physical sector of the Jordan--Schwinger transformation at $r=0$.
\par
The family \eqref{eq:symmetric_Splus_r}--\eqref{eq:symmetric_Sz_r}
constructed above shows that the standard two-mode Jordan--Schwinger representation
does not exhaust all two-mode bosonic realizations of the
$\mathfrak{su}(2)$ algebra compatible with a fixed choice of the operator $\hat S^z$.
At $r=0$, it reduces to the ordinary Jordan--Schwinger representation,
whereas for $r\neq 0$ it yields a nontrivial deformation in which the conserved operator is no longer the total boson number $\hat N$
but its squeezed image $\hat N_r$, defined in \eqref{eq:Nr_def}.
Thus, the distinction between the standard Jordan--Schwinger representation
and the family \eqref{eq:symmetric_Splus_r}--\eqref{eq:symmetric_Sz_r}
lies not in a modification of the $\mathfrak{su}(2)$ algebra itself, but in the manner in which
its irreducible spin blocks are embedded into the two-mode bosonic
Fock space. In the standard case, these blocks are selected by the condition
$\hat N=2S$, whereas in the present family they are
defined by $\hat N_r = 2S$.

It should be emphasized that the family
\eqref{eq:symmetric_Splus_r}--\eqref{eq:symmetric_Sz_r}
does not exhaust all two-mode realizations of $\mathfrak{su}(2)$
that fail to conserve $\hat N$.
It is singled out here as a minimal nontrivial yet fully
controlled extension of the standard Jordan--Schwinger ansatz.
A more general class of solutions may include coefficients that depend
on $\hat n_1$ and $\hat n_2$, channels with $\Delta N=\pm4,\pm6,\dots$,
and conserved quantities other than $\hat N_r$. Thus, the family \eqref{eq:symmetric_Splus_r}--\eqref{eq:symmetric_Sz_r}
should be regarded not as the most general solution, but as a minimal
example of an exact two-mode representation beyond the standard
$N$-conserving Jordan--Schwinger class.

\selectlanguage{english}
\section{Conclusion}
\par In this work, we have developed a unified algebraic approach to spin--boson transformations based on the formalism of $f$-deformed oscillators. This approach cleanly separates the part of the solution uniquely fixed by the commutation relations of the $\mathfrak{su}(2)$ algebra from the factorization freedom that specifies a particular realization. In the single-mode case, we have shown that the Holstein--Primakoff and Dyson--Maleev transformations and the interpolating $\alpha$-family arise as different factorizations of a single algebraically determined object. The standard spin--boson mappings can therefore be interpreted as realizations of a common underlying structure, making it possible to distinguish the exact algebraic content on the physical subspace from effects associated with non-Hermiticity, the choice of metric, and off-block extensions.

In the two-mode case, the same approach yields both deformed generalizations of the Jordan--Schwinger representation and exact examples of two-mode realizations beyond the standard $N$-conserving class. The method therefore serves not only as a means of unifying known transformations but also as a tool for the systematic construction of new bosonic representations.

A central result is the explicit isolation of the invariant algebraic part of the solution and its separation from the additional structures associated with the choice of factorization, metric, and operator extension beyond the physical subspace. This distinction naturally connects the standard spin--boson transformations with questions of quasi-Hermiticity, truncation, numerical stability, and discrete expansions of operator functions.

The scope of the proposed approach is not limited to the $\mathfrak{su}(2)$ algebra: it also provides a general method for constructing bosonic realizations and non-Hermitian representations of a broader class of algebraic structures.

\section*{Conflict of Interest}
The authors declare no conflict of interest.

\selectlanguage{english}
\section*{Acknowledgments}
The authors thank their colleagues for helpful discussions. This research received no specific funding.

\bibliographystyle{unsrtnat}
\bibliography{bib_TMF}

\appendix

\section{Shift Identities}
\label{app:shift_domains}
\par Let $\hat a$ and $\hat a^\dagger$ be canonical bosonic operators satisfying $[\hat a,\hat a^\dagger]=1$, and let $\hat n=\hat a^\dagger\hat a$ be the number operator.
For any function $F:\mathbb N_0\to\mathbb C$, define the operator $F(\hat n)$ spectrally by
$F(\hat n)\ket{n}=F(n)\ket{n}$.
When necessary, we set $F(-1)=0$ so that the corresponding shifted expressions are well defined at the boundary of the spectrum. The following identities then hold:
\begin{equation}
\hat a\,F(\hat n)=F(\hat n+1)\,\hat a,\qquad
F(\hat n)\,\hat a^\dagger=\hat a^\dagger\,F(\hat n+1),
\label{eq:shift_forward}
\end{equation}
and, equivalently,
\begin{equation}
F(\hat n)\,\hat a=\hat a\,F(\hat n-1),\qquad
\hat a^\dagger\,F(\hat n)=F(\hat n-1)\,\hat a^\dagger.
\label{eq:shift_backward}
\end{equation}
\begin{proof}
It suffices to verify Eq.~\eqref{eq:shift_forward} on the basis vectors $\ket{n}$.
Using $\hat a\ket{n}=\sqrt n\,\ket{n-1}$, we obtain
\begin{align}
\hat a\,F(\hat n)\ket{n}
=
\sqrt n\,F(n)\ket{n-1}.
\end{align}
On the other hand,
\begin{align}
F(\hat n+1)\,\hat a\ket{n}
=
F(n)\sqrt n\,\ket{n-1}.
\end{align}
Therefore, $\hat aF(\hat n)=F(\hat n+1)\hat a$ for every basis vector in $\{\ket{n}\}$.
The second identity in \eqref{eq:shift_forward} is proved analogously using
$\hat a^\dagger\ket{n}=\sqrt{n+1}\ket{n+1}$.
The backward-shift identities \eqref{eq:shift_backward} follow from \eqref{eq:shift_forward} by shifting the argument of the function.\footnote{Equation~\eqref{eq:shift_forward} may be regarded as a natural generalization of the standard commutator identity for the number operator. In particular, for $F(\hat n)=\hat n=\hat a^\dagger\hat a$, one has
$\hat a\,\hat n=\hat a\hat a^\dagger\hat a=(\hat a^\dagger\hat a+1)\hat a=(\hat n+1)\hat a$.}
\end{proof}
\section{Solution of the Universal Finite-Difference Equation}
\label{app:diff_eq_solution}
For the single-mode ansatz \eqref{eq:universal_one_mode_ansatz_ch3}, the commutator condition of the $\mathfrak{su}(2)$ algebra reduces to Eq.~\eqref{eq:diff_eq_operator_ch3}, which takes the following form in the basis $\ket{n}$:
\begin{equation}
(n+1)G(n+1)-nG(n)=2(S-n),\qquad n\in\mathbb{N}_0.
\label{eq:diff_eq_scalar_ch3}
\end{equation}
Define $Q(n)\equiv n G(n)$, so that $Q(0)=0$ independently of $G(0)$. The recurrence relation can then be written as
\begin{equation}
Q(n+1)-Q(n)=2(S-n).
\label{eq:Q_rec_app}
\end{equation}
Summing the recurrence \eqref{eq:Q_rec_app} over $k$ from $0$ to $n-1$, we obtain
\begin{equation}
Q(n)=\sum_{k=0}^{n-1}2(S-k)=2nS-n(n-1).
\label{eq:Q_closed_app}
\end{equation}
Hence, for $n\ge 1$,
\begin{equation}
G(n)=\frac{Q(n)}{n}=2S-(n-1).
\label{eq:G_n_ge1_app}
\end{equation}
The value of $G(0)$ remains undetermined because it always appears multiplied by $n$. A convenient convention, consistent with the smooth continuation of \eqref{eq:G_n_ge1_app}, is
\begin{equation}
G(0)=2S+1,
\label{eq:G0_app}
\end{equation}
as stated below Eq.~\eqref{eq:G_solution_ch3}.

\section{Newton Series: Finite-Difference Normal Ordering on a Finite Block}
\label{app:newton_series}
Finite-difference (Newton-series) expansions are naturally adapted to operator functions defined on the discrete spectrum
$n\in\mathbb{N}_0$ and provide a direct route to normal ordering.
For a function $F:\mathbb{N}_0\to\mathbb{C}$, define the forward-difference operator by
\begin{equation}
(\Delta F)(n)=F(n+1)-F(n),\qquad \Delta^k\equiv\Delta\circ\cdots\circ\Delta.
\label{eq:Delta_def_app}
\end{equation}
The Newton expansion about $n=0$ then reads
\begin{equation}
F(\hat n)=\sum_{k=0}^{\infty}\frac{\Delta^kF(0)}{k!}\,\hat n^{(k)},
\quad
\hat n^{(k)}\equiv\hat n(\hat n-1)\cdots(\hat n-k+1),
\label{eq:newton_series_app}
\end{equation}
where the coefficients are given by the binomial transform
\begin{equation}
\Delta^kF(0)=\sum_{\ell=0}^{k}(-1)^{k-\ell}\binom{k}{\ell}\,F(\ell).
\label{eq:binomial_transform_app}
\end{equation}
The key identity, which follows directly by applying both sides to $\ket{n}$ and using
$\hat a\ket{n}=\sqrt{n}\ket{n-1}$, is
\begin{equation}
\hat n^{(k)} = :\hat n^k: = \hat a^{\dagger k}\hat a^{k},
\label{eq:falling_factorial_normal_order_app}
\end{equation}
so that the expansion \eqref{eq:newton_series_app} is already normally ordered term by term.

On the finite set $n\in\{0,1,\dots,2S\}$, any function $F(n)$ admits an exact representation
by a unique polynomial of degree $\le 2S$ through interpolation~\cite{konig2021newtonseriesexpansionbosonic,vogl_resummation_2020}.
In the Newton basis, this corresponds to the finite expansion
\begin{equation}
F(\hat n)\Big|_{\mathcal{H}_{\mathrm{phys}}}
=\sum_{k=0}^{2S}\frac{\Delta^kF(0)}{k!}\,\hat a^{\dagger k}\hat a^{k},
\label{eq:newton_finite_block_app}
\end{equation}
which reproduces the values of $F(n)$ exactly for all
$0\le n\le 2S$. Hence, when the corresponding spin generators are
defined on the finite physical subspace $\mathcal H_{\mathrm{phys}}$,
the Newton representation reproduces the required operator functions
exactly and introduces no interpolation error. The corresponding
$\mathfrak{su}(2)$ commutation relations are therefore recovered
exactly on this subspace.
This distinguishes the finite Newton interpolation from truncated Taylor expansions, which are not adapted to the discrete spectrum and generally fail to account for the boundary conditions at $n=0$ and $n=2S$.

Exactness on the physical spectrum does not, however, imply a global
operator identity on the full Fock space. In general, the interpolating
polynomial does not coincide with the original function for $n>2S$,
and the corresponding polynomial operator identities need not hold
outside the physical subspace.

This distinction is important in view of known restrictions on
polynomial realizations of semisimple Lie algebras in quantum canonical
operators and, in particular, no-go results for global single-mode
realizations of compact $\mathfrak{su}(2)$ by skew-Hermitian polynomial
elements of the canonical Weyl algebra~\cite{Joseph1972,Heib2025}.
The standard spin--boson transformations considered here do not conflict
with these results. The Holstein--Primakoff transformation contains the
nonpolynomial spectral function $\sqrt{2S-\hat n}$, whereas the
Dyson--Maleev transformation, although polynomial, does not satisfy the
adjointness condition
$\hat S^-=(\hat S^+)^\dagger$
with respect to the standard Fock-space inner product. In the latter
case, the physical adjointness structure is recovered on the finite
physical space through the corresponding positive metric operator.
Likewise, finite polynomial representatives obtained by Newton
interpolation are exact on the specified physical spin space, but this
does not imply the existence of the same global polynomial realization
of compact $\mathfrak{su}(2)$ on the full single-mode Fock space.

\subsection{\texorpdfstring{Application to the \(\alpha\)-Family}
{Application to the alpha-Family}}

For convenience, we use here the unnormalized representative of the
same factorization class. For the present application, only the values
of the deformation functions on the physical subspace are required.
Define
\begin{equation}
F_\alpha(n)\equiv(2S+1-n)^\alpha,
\qquad n=0,1,\ldots,2S.
\label{eq:Falpha_def_app}
\end{equation}
Up to the constant normalization and the shift $n\mapsto n-1$
already incorporated in
\eqref{eq:alpha_family_Splus_ch3}--\eqref{eq:alpha_family_Sminus_ch3},
the corresponding unnormalized deformation factors are
$F_\alpha(\hat n)$ and $F_{1-\alpha}(\hat n)$.

Using \eqref{eq:newton_series_app}--\eqref{eq:newton_finite_block_app},
we obtain the exact normally ordered representation on the physical subspace,
\begin{align}
F_\alpha(\hat n)\Big|_{\mathcal H_{\mathrm{phys}}}
&=
\sum_{k=0}^{2S}
\frac{\Delta^kF_\alpha(0)}{k!}
\hat a^{\dagger k}\hat a^k,
\\
\Delta^kF_\alpha(0)
&=
\sum_{\ell=0}^{k}
(-1)^{k-\ell}\binom{k}{\ell}(2S+1-\ell)^\alpha.
\label{eq:Falpha_newton_coeffs_app}
\end{align}
On the physical subspace, i.e., after projection onto
$\mathcal H_{\mathrm{phys}}$, it follows that
\begin{align}
\hat S_{(\alpha)}^+
&=\hat a\sum_{k=0}^{2S}\frac{\Delta^kF_\alpha(0)}{k!}\,\hat a^{\dagger k}\hat a^{k},\\[1mm]
\hat S_{(\alpha)}^-
&=\left(\sum_{k=0}^{2S}\frac{\Delta^kF_{1-\alpha}(0)}{k!}\,\hat a^{\dagger k}\hat a^{k}\right)\hat a^\dagger.
\label{eq:Salpha_normal_order_app}
\end{align}


\section{Equivalence of Left-, Center-, and Right-Ordered Conventions}
\label{app:JS_local_forms}
For bosonic operators and an operator-valued function $F(\hat n_1,\hat n_2)$, the standard shift identities read
\begin{align}
\hat a_1^\dagger F(\hat n_1,\hat n_2)
&=
F(\hat n_1-1,\hat n_2)\hat a_1^\dagger, \\
F(\hat n_1,\hat n_2)\hat a_1
&=
\hat a_1 F(\hat n_1-1,\hat n_2),\nonumber
\\
\hat a_2^\dagger F(\hat n_1,\hat n_2)
&=
F(\hat n_1,\hat n_2-1)\hat a_2^\dagger, \nonumber\\
F(\hat n_1,\hat n_2)\hat a_2
&=
\hat a_2 F(\hat n_1,\hat n_2-1).\nonumber
\end{align}
These relations show that the left-, center-, and right-ordered conventions
are equivalent and differ only by a relabeling of the function arguments. In particular,
\begin{align}
\hat a_1^\dagger F(\hat n_1,\hat n_2)\hat a_2
&=
\hat a_1^\dagger \hat a_2\,F(\hat n_1,\hat n_2-1),\\\nonumber
F(\hat n_1,\hat n_2)\hat a_1^\dagger\hat a_2
&=
\hat a_1^\dagger \hat a_2\,F(\hat n_1+1,\hat n_2-1).
\end{align}
Thus, adopting the right-ordered convention entails no loss of generality.

\section{\texorpdfstring
{Derivation of Equation~\eqref{eq:lattice_constraint}}
{Derivation of the Lattice Constraint}}
\label{app:JS_lattice_constraint_derivation}
Equation~\eqref{eq:JS_right_local} directly gives
\begin{align}
\hat S^+\ket{n_1,n_2}
&=
\sqrt{(n_1+1)n_2}\,f(n_1,n_2)\ket{n_1+1,n_2-1},
\\
\hat S^-\ket{n_1,n_2}
&=
\sqrt{n_1(n_2+1)}\,g(n_1,n_2)\ket{n_1-1,n_2+1},
\end{align}
where, by convention, the vector is taken to be zero if either index
becomes negative. Successive application of the ladder operators yields
\begin{align} &\hat S^+\hat S^-\ket{n_1,n_2} = \nonumber\\ &n_1(n_2+1)\,g(n_1,n_2)f(n_1-1,n_2+1)\ket{n_1,n_2}, \\
&\hat S^-\hat S^+\ket{n_1,n_2} = \nonumber\\ &n_2(n_1+1)\,f(n_1,n_2)g(n_1+1,n_2-1)\ket{n_1,n_2}. \end{align}
Using the definition \eqref{eq:H_def}, we obtain
\begin{align}
\hat S^+\hat S^-\ket{n_1,n_2}
&=
n_1(n_2+1)\,H(n_1-1,n_2+1)\ket{n_1,n_2},
\\
\hat S^-\hat S^+\ket{n_1,n_2}
&=
n_2(n_1+1)\,H(n_1,n_2)\ket{n_1,n_2}.
\end{align}
Therefore, the condition
\begin{align}
[\hat S^+,\hat S^-]\ket{n_1,n_2}
=
(n_1-n_2)\ket{n_1,n_2}
\end{align}
is equivalent to Eq.~\eqref{eq:lattice_constraint}.

\section{Reduction to the Single-Mode \texorpdfstring{$\alpha$}{alpha}-Family}
\label{app:JS_alpha_reduction_ru}
As shown in Sec.~\ref{sec:two_mode_js}, on a fixed block $\mathcal H_{2S}$, after making the identification
$\ket{n}_S \equiv \ket{2S-n,n}$ and imposing
$\hat S^z\ket{n}_S=(S-n)\ket{n}_S$,
the $\mathfrak{su}(2)$ algebra fixes only the product
$f_{n+1}g_n=1$,
$n=0,1,\dots,2S-1$. For convenience, we now pass from the unnormalized choice \eqref{eq:fg_alpha_two_mode} to the standard ``physical'' normalization by the reciprocal rescaling $f_n\mapsto(2S)^{\frac12-\alpha}f_n$ and $g_n\mapsto(2S)^{\alpha-\frac12}g_n$, which leaves the invariant product $f_{n+1}g_n=1$ unchanged:
\begin{equation}
\!\!\!f_n^{(\alpha)}
=
\left(1-\frac{n-1}{2S}\right)^{\alpha-\frac12}\!\!\!\!\!\!,
\,
g_n^{(\alpha)}
=
\left(1-\frac{n}{2S}\right)^{\frac12-\alpha}\!\!\!\!\!\!\!\!\!,
\,
\alpha\in\mathbb R.
\label{eq:fg_alpha_block_ru}
\end{equation}
Here and below, only the values entering the action of $\hat S^\pm$ on $\mathcal H_{2S}$ are required: $f_n^{(\alpha)}$ for $n=1,\ldots,2S$ and $g_n^{(\alpha)}$ for $n=0,\ldots,2S-1$.
Indeed,
\begin{align}
f_{n+1}^{(\alpha)}g_n^{(\alpha)}
=
\left(1-\frac{n}{2S}\right)^{\alpha-\frac12}
\left(1-\frac{n}{2S}\right)^{\frac12-\alpha}
=1.
\end{align}
Accordingly, the two-mode ladder operators act on $\mathcal H_{2S}$ as
\begin{align}
\hat S^+_{\alpha}\ket{n}_S
&=
\sqrt{(2S-n+1)n}\,f_n^{(\alpha)}\ket{n-1}_S
\nonumber\\
&=
\sqrt{2S}\,\sqrt n
\left(1-\frac{n-1}{2S}\right)^\alpha
\ket{n-1}_S,
\label{eq:Jplus_alpha_block_ru}
\\[1mm]
\hat S^-_{\alpha}\ket{n}_S
&=
\sqrt{(n+1)(2S-n)}\,g_n^{(\alpha)}\ket{n+1}_S
\nonumber\\
&=
\sqrt{2S}\,\sqrt{n+1}
\left(1-\frac{n}{2S}\right)^{1-\alpha}
\ket{n+1}_S.
\label{eq:Jminus_alpha_block_ru}
\end{align}
Using the shift identities, these relations are equivalent to the single-mode expressions
\begin{align}
\hat S^+_{\alpha}
&=
\sqrt{2S}\,\hat a
\left(1-\frac{\hat n-1}{2S}\right)^\alpha
\equiv
\sqrt{2S}\left(1-\frac{\hat n}{2S}\right)^\alpha\hat a,
\label{eq:Jplus_to_one_mode_ru}
\\[1mm]
\hat S^-_{\alpha}
&=
\sqrt{2S}\,\hat a^\dagger
\left(1-\frac{\hat n}{2S}\right)^{1-\alpha}
\equiv
\sqrt{2S}\left(1-\frac{\hat n-1}{2S}\right)^{1-\alpha}\hat a^\dagger.
\label{eq:Jminus_to_one_mode_ru}
\end{align}
Thus, under the identification $\mathcal H_{2S}\cong\mathcal H_{\mathrm{phys}}$,
the two-mode $\alpha$-family \eqref{eq:Jplus_alpha_block_ru}--\eqref{eq:Jminus_alpha_block_ru}
reduces exactly to the single-mode $\alpha$-family of Sec.~\ref{sec:one_mode_diff}. Consequently, on a fixed block $N=2S$, it does not define a new independent structure but coincides exactly with the single-mode case.
\section{Weight of Operator Monomials}
\label{app:weight_monomials}
We say that an operator $\hat X$ has weight $k$ if
$[\hat S^z,\hat X]=k\,\hat X$. Let $F(\hat n_1,\hat n_2)$ be an operator-valued function defined spectrally
on $\mathbb N_0^2$. Then, for the normally ordered monomial
\begin{equation}
\hat M_{pqrs}[F]
\equiv
(\hat a_1^\dagger)^p
\hat a_1^q
(\hat a_2^\dagger)^r
\hat a_2^s
\,F(\hat n_1,\hat n_2),
\label{eq:general_monomial_weight}
\end{equation}
where $p,q,r,s\in\mathbb N_0$, one has
\begin{equation}
[\hat S^z,\hat M_{pqrs}[F]]
=
\frac{p-q-r+s}{2}\,
\hat M_{pqrs}[F].
\label{eq:general_monomial_weight_result}
\end{equation}
In other words, the weight of the monomial \eqref{eq:general_monomial_weight} is
$\tfrac{p-q-r+s}{2}$.
\begin{proof}
Since $F(\hat n_1,\hat n_2)$ depends only on $\hat n_1$ and $\hat n_2$, one has
\begin{equation}
[\hat S^z,F(\hat n_1,\hat n_2)]=0.
\label{eq:F_commutes_Sz}
\end{equation}
Furthermore, the canonical commutation relations
\begin{eqnarray}
[\hat n_1,\hat a_1^\dagger]=\hat a_1^\dagger,
\qquad
[\hat n_1,\hat a_1]=-\hat a_1,
\end{eqnarray}
\begin{eqnarray}
[\hat n_2,\hat a_2^\dagger]=\hat a_2^\dagger,
\qquad
[\hat n_2,\hat a_2]=-\hat a_2
\end{eqnarray}
together with the mutual commutativity of operators belonging to different modes imply
\begin{align}
[\hat S^z,\hat a_1^\dagger]&=\frac12\hat a_1^\dagger,
&
[\hat S^z,\hat a_1]&=-\frac12\hat a_1,
\label{eq:Sz_a1_weights}
\\
[\hat S^z,\hat a_2^\dagger]&=-\frac12\hat a_2^\dagger,
&
[\hat S^z,\hat a_2]&=\frac12\hat a_2.
\label{eq:Sz_a2_weights}
\end{align}
We now apply the Leibniz rule
\begin{align}
[\hat S^z,\hat X\hat Y]=[\hat S^z,\hat X]\hat Y+\hat X[\hat S^z,\hat Y].
\end{align}
Applying it successively to the product
$(\hat a_1^\dagger)^p \hat a_1^q (\hat a_2^\dagger)^r \hat a_2^s F(\hat n_1,\hat n_2)$
and using \eqref{eq:F_commutes_Sz}, \eqref{eq:Sz_a1_weights}, and
\eqref{eq:Sz_a2_weights}, we obtain
\begin{align*}
&[\hat S^z,\hat M_{pqrs}[F]]
=
\left(
\frac{p}{2}
-\frac{q}{2}
-\frac{r}{2}
+\frac{s}{2}
\right)
(\hat a_1^\dagger)^p
\hat a_1^q\\\nonumber
&\times
(\hat a_2^\dagger)^r
\hat a_2^s
F(\hat n_1,\hat n_2)
=
\frac{p-q-r+s}{2}\,\hat M_{pqrs}[F].
\end{align*}
\end{proof}
For the ansatz with $\Delta N=0,\pm2$, the monomial
$\hat a_1^\dagger\hat a_2$ has weight
\[
\frac{1-0-0+1}{2}=1.
\]
Similarly, the monomials $(\hat a_1^\dagger)^2$ and $\hat a_2^2$ both have weight $1$,
so all three terms in \eqref{eq:minimal_extended_Splus}
have weight $+1$, and hence
$[\hat S^z,\hat S^+]=\hat S^+$. Likewise,
$\hat a_2^\dagger\hat a_1$, $\hat a_1^2$, and $(\hat a_2^\dagger)^2$
have weight $-1$, so that $[\hat S^z,\hat S^-]=-\hat S^-$. Therefore, the relations $[\hat S^z,\hat S^\pm]=\pm \hat S^\pm$ hold automatically, and the only nontrivial condition imposed by the $\mathfrak{su}(2)$ algebra is the commutator
$[\hat S^+,\hat S^-]=2\hat S^z$.

\begin{theorem}
\label{thm:constant_subclass_solution}
The operators \eqref{eq:constant_Splus}--\eqref{eq:constant_Sminus}
furnish an algebraic realization of the $\mathfrak{su}(2)$
commutation relations, i.e., they satisfy
\begin{align}
[\hat S^z,\hat S^\pm]=\pm \hat S^\pm,
\qquad
[\hat S^+,\hat S^-]=2\hat S^z,
\end{align}
if and only if the coefficients $a,b,c,d,e,h$ satisfy the system
\begin{equation}
ah=bc,\quad
bd=ae,\quad
ce=dh,\quad
ab-4ce=1.
\label{eq:constant_constraints_1}
\end{equation}
\end{theorem}

\begin{proof}
As noted above, the relations $[\hat S^z,\hat S^\pm]=\pm \hat S^\pm$
hold automatically. In the $\Delta N=+2$ channel, equating the coefficient of the state $\ket{n_1+1,n_2+1}$
in the action of the commutator $[\hat S^+,\hat S^-]\ket{n_1,n_2}$ gives
\begin{equation}
(n_2+2)\,ah+n_1\,bc=(n_1+2)\,bc+n_2\,ah.
\label{eq:deltaNplus2_constant}
\end{equation}
Since this equality must hold for all $n_1,n_2\in\mathbb N_0$, we obtain
\begin{equation}
ah=bc.
\label{eq:constraint_ah_bc}
\end{equation}
Similarly, equating the coefficient of the state $\ket{n_1-1,n_2-1}$ yields
\begin{equation}
(n_1-1)\,ae+(n_2+1)\,db=(n_1+1)\,ae+(n_2-1)\,db,
\label{eq:deltaNminus2_constant}
\end{equation}
which implies
\begin{equation}
bd=ae.
\label{eq:constraint_bd_ae}
\end{equation}
Equating the coefficient of $\ket{n_1,n_2}$ gives
\begin{eqnarray}
&n_1(n_2+1)\,ab+n_1(n_1-1)\,ce+(n_2+1)(n_2+2)\,dh
\nonumber\\
&-n_2(n_1+1)\,ab-(n_1+1)(n_1+2)\,ce-n_2(n_2-1)\,dh
\nonumber\\
&=
n_1-n_2.
\label{eq:diagonal_constant_raw}
\end{eqnarray}
After simplification, we obtain
\begin{equation}
ab(n_1-n_2)- (4n_1+2)\,ce + (4n_2+2)\,dh = n_1-n_2.
\label{eq:diagonal_constant_simplified}
\end{equation}
Since this equality must hold for all $n_1,n_2$, comparison of the coefficients of $n_1$, $n_2$, and the constant term gives
\begin{equation}
ab-4ce=1,\qquad
ab-4dh=1,\qquad
ce=dh.
\label{eq:diagonal_constant_constraints}
\end{equation}
The last two equalities are equivalently expressed as
\[
ce=dh,
\qquad
ab-4ce=1.
\]
Together with \eqref{eq:constraint_ah_bc} and \eqref{eq:constraint_bd_ae},
these relations yield the system \eqref{eq:constant_constraints_1}.
\end{proof}

\end{document}